\documentclass[12pt]{article}

\usepackage{amssymb}
\usepackage{amsmath}
\usepackage{amsthm}
\usepackage{graphicx}
\usepackage{xcolor}
\usepackage{natbib}
\usepackage{booktabs}
\usepackage{makecell}
\usepackage{url}
\usepackage[T1]{fontenc}
\usepackage{enumitem}
\usepackage{cancel}
\usepackage[normalem]{ulem}
\usepackage{float}
\usepackage{tikz}

\theoremstyle{plain}

\newtheorem{proposition}{Proposition}

\newtheorem{theorem}{Theorem}

\theoremstyle{definition}
\newtheorem{example}{Example}
\newtheorem{definition}{Definition}

\theoremstyle{remark}
\newtheorem{remark}{Remark}

\newcommand{\domain}{\bigcup_{\mathcal{D}\subseteq\mathfrak{X}}\mathcal{R}(\mathcal{D})}
\newcommand{\codomain}{\mathcal{B}(X)}
\newcommand{\sign}{\operatorname{sign}}
\newcommand{\topsum}{\mathbin{\mid \joinrel =}}
\newcommand{\bottomsum}{\mathbin{= \joinrel \mid}}
\newcommand{\He}{\sign}

\title{Consistencies in Social Ranking}

\author{
Takahiro Suzuki\thanks{Corresponding author. Email: suzuki-tkenmgt@g.ecc.u-tokyo.ac.jp} \\
Department of Civil Engineering \\
Graduate School of Engineering \\
The University of Tokyo \\
Tokyo 113-8656, Japan
\and
Michele Aleandri \\
LUISS University \\
Viale Romania 32 \\
00197 Rome, Italy
\and
Stefano Moretti \\
LAMSADE, CNRS, Université Paris-Dauphine, Université PSL \\
75016 Paris, France
}

\date{}

\begin{document}
\maketitle



\begin{abstract}
Ranking individuals based on their performance in different coalitions is a problem emerging in various domains (teams sports, scientific evaluation, argumentation, etc.). Often, for practical reasons, the number of comparable coalitions is limited. Therefore, the foundational principles of ranking solutions must support realistic interpretations in contexts where only certain coalitions can be compared. To address this issue, in this paper we present an axiomatic analysis of solutions for the social ranking problem centered on the notion of {\it consistency}. More precisely, we show that an appropriate notion of consistency, which specifies how to combine rankings on individuals across different rankings on coalitions, plays a key role in any axiomatic characterization, representing the true distinguishing feature of each solution. This role is further highlighted by the taxonomy of the complementary axioms used in our characterizations, which boil down to well-studied properties of invariance with respect to the label of players or coalitions, and also with respect to minor changes in a coalitional ranking. By showing the logical independence of the axioms used in each characterization, as well as a rigorous analysis of alternative notions of consistency with respect to the majority of solutions from the literature, this work attempts to provide a first systematic study of the social ranking problem over a variable domain of coalitions.
\end{abstract}




\section{Introduction}
In many real-life situations, as in the ranking of team sports players or in the process to select best researchers,  individuals are ranked based on their role in coalitions (subsets of individuals): football players, for instance, are compared considering their collaborative performances on different lineups, and researchers are compared looking at their joint works with multiple co-authors.
Thus, we require formal methods to convert the  information about coalitions' comparison into an ordinal comparison between individuals. To this purpose, a \textit{social ranking solution} (SRS) is a function that maps a ranking of coalitions to a ranking of individuals keeping into account their synergies. The notion of SRS was proposed in \cite{Moretti2017} and a number of subsequent studies have been conducted in the last decade (we refer to the ``Related literature'' section for a short introduction).
A technical difficulty regarding the study of SRSs lies within the abundance of possible coalitions: there are $2^{n}$ logically possible coalitions (subsets) when there are $n$ individuals. This causes not only the computational difficulty of the analysis, but also the difficulty of even obtaining the input data (i.e., comparing all the logically possible coalitions is unrealistic in most cases). 

In this view, we will study SRSs with a variable domain of coalitions---that is, the set of all feasible coalitions can vary. A key axiom in this context is \textit{consistency} and 
its study for SRSs can be traced back to \cite{Suzuki2024}. Roughly speaking, this property requires that the ranking of individuals must not depend on how the coalitional rankings over disjoint families of subsets of the individuals’ set are combined; thus, we can reasonably focus on the partial information. For instance, in team sports, a technical director of a team might face a problem when combining a ranking of lineups based on their performance in a national league with a ranking of alternative lineups performing in an international competition; the director of a laboratory might have information to compare groups of researchers according to a criterion linked to scientific production (e.g., co-authored papers), but some papers are co-authored by a few authors, and others by large collaboration groups;
or a streaming platform might be interested in comparing the success of bunches of movies proposed by competing streaming platforms to finally improve its offer of films, etc. In these situations, even if the original population on which one wants to establish a ranking is the same,  the use of evaluation criteria that are incommensurable with each other makes some coalitions hardly comparable. 

As a more detailed example, consider the evaluation of two hypothetical researchers, referred to as $a$ and $b$, who contributed to scientific publications within two distinct small collaborative groups, $S_a$ and $S_b$, respectively, as well as within large-scale collaborations with many co-authors, denoted as $L_a$ and $L_b$. In this context, the evaluation of their individual contributions must take into account the different nature of authorship attribution dynamics: within $S_a$ and $S_b$, the intellectual contribution  of roles can be more easily attributed to each author, given the limited number of collaborators. Conversely, within $L_a$ and $L_b$, the evaluation of of $a$ and $b$'s contributions should focus on more complex criteria, such as theoretical and methodological skills, research leadership, etc., rather than relying solely on authorship position, which may not reflect the nuanced nature of authors' involvement in large collaborative efforts. Therefore, distinct rankings could be established between $a$ and $b$, depending on whether one considers the performance of $S_a$ {\it vs.} $S_b$, rather than that of $L_a$ {\it vs.} $L_b$.
However, the publication performance of the four teams can be compared on a well-defined evaluation dimension, such as journal quality or publication impact, and thus an overall ranking between the groups can also be established.
Alternative definitions of consistency properties identify how partial rankings on coalitions of similar size should extend to situations where all coalitions are compared to each other, independently on their size.
Suppose, for instance, that due to the quality of the journals where their  papers have been published, the publication performance of group $S_a$ is not worse than the one of group $S_b$ (in short, $S_a \succsim_1 S_b$), and  the publication performance of group $L_a$ is not worse than the one of group $L_b$ (in short, $L_a \succsim_2 L_b$). It seems reasonable to assume that, if one believes that researcher $a$ has greater merit than researcher $b$ for both coalitional rankings $\succsim_1$ and $\succsim_2$   then, for any ranking $\succsim$ constructed as a  compatible combination of $\succsim_1$ and $\succsim_2$ (e.g., $S_a \succsim S_b \succsim L_a \succsim L_b$, or $L_a \succsim S_a \succsim S_b \succsim L_b$, or $S_a \succsim L_a \succsim L_b \succsim S_b$, etc.) $a$ should also be regarded as having greater merit than $b$.

More generally, let $\mathcal{D}$ be a set of coalitions of a finite set $X$ of individuals, and $\mathcal{D}_{1}$ and $\mathcal{D}_{2}$ be a partition of $\mathcal{D}$.
Assume that the performance ranking of coalitions in $\mathcal{D}$ is $\succsim$; $\succsim_{i}$ is the restriction of $\succsim$ to $\mathcal{D}_{i}$ for each $i=1,2$ (put the other way around, $\succsim$ is a ranking on $\mathcal{D}$ compatible with both $\succsim_{1}$ and $\succsim_{2}$). Then, consistency requires that for any two individuals $x$ and $y$, if the social ranking of $x$ is at least as good as that of $y$ at both $\succsim_{1}$ and $\succsim_{2}$, then so should it be at $\succsim$; furthermore, in such a situation, if the social ranking is strict at either $\succsim_{1}$ or $\succsim_{2}$, then the social ranking of $x$ must be strictly better than that of $y$ at $\succsim$. In this sense, consistency requires that the social ranking obtained from the partial information in $\succsim_{1}$ and $\succsim_{2}$ should be in agreement with the social ranking obtained from the whole information in their combination ($\succsim$). 

Clearly, the consistency principle is logically demanding due to the abundance of combinations compatible with the two rankings $\succsim_{1}$ and $\succsim_{2}$. Therefore, alternative approaches can be adopted in contexts characterized by information that 
could significantly limit the number of compatible combinations. 
For example, the notion of {\it concatenation consistency} deals with situations where the two rankings $\succsim_{1}$ and $\succsim_{2}$ are combined to reflect a hierarchy of priorities over the criteria used to rank coalitions. 
Following the previous example with small and large groups of authors, the extension ranking  $\succsim$ according to the concatenation sum is the one reflecting a higher priority of the ranking on small groups over the one on larger groups, so that
$\succsim_{1}$ and $\succsim_{2}$ are juxtaposed to produce  $S_a \succsim S_b \succ L_a \succsim L_b$.
Following an alternative direction aimed at maintaining a common structure of equivalence classes between the two rankings $\succsim_{1}$ and $\succsim_{2}$, we investigate other notions of sum where equivalence classes are merged either starting from the best coalitions ({\it top-aligned sum}) or from the worst ones ({\it bottom-aligned sum}). So, in the previous example,  if $S_a \sim_1 S_b$ (the two small groups are indifferent in terms of publication performance)  and $L_a \succ_2 L_b$ (group $L_a$ performs strictly stronger  than $L_b$)
 the top-aligned sum is $S_a \sim S_b \sim L_a \succ L_b$, while the bottom-aligned sum is $ L_a \succ S_a \sim S_b \sim L_b$.

 
 As already mentioned, the choice of an appropriate definition of sum is context-dependent and might reflect how criteria for coalitions performance are compared. Nevertheless, each corresponding notion of consistency relies on the same rules to combine the outcomes of an SRS over two rankings, no matter which notion of sum is adopted: if an individual $x$ is socially not worse than an individual $y$ on both $\succsim_{1}$ and $\succsim_{2}$, it must be so even on the sum $\succsim$ (see Section \ref{axioms_consist} for more details).
In this paper, we study the effects of these alternative notions of consistency for SRSs. We show that each of these axioms plays a key role in characterizing a particular SRS; this is together with other families of complementary axioms that are grouped into two main families:  axioms invariant to some permutations of individuals or coalitions, and axioms involving a notion of independence of minor modifications in the coalitional ranking  (see  Section \ref{sec:compaxioms} for more details). Concatenation consistency is crucial for the axiomatic characterization of the Sign lex-cel solution (Theorem \ref{theorem:SL}), a discrete version of lex-cel \cite{Bernardi2019}. Instead, top-aligned consistency is fundamental for axiomatizing Plurality (Theorem \ref{theorem:plurality}), another SRS recently introduced in \cite{Suzuki2025_Mill_SRS}.
Finally, the lex-cel SRS is characterized by the combination of multiple notions of consistency.
Overall, the three axiomatic characterizations we propose and study in this paper show the existence of a common block of properties satisfied by the three distinct SRSs, as well as a peculiar role played by alternative notions of consistency that we believe represent the distinctive features of each SRS studied in this work.



The roadmap of the paper is as follows. 
We start with the basic model description in Section \ref{section:social ranking solutions}. Consistency axioms are introduced in Section \ref{axioms_consist}. After we provide some preliminary result (subsection \ref{subsection:relations among consistency axioms}) and complementary axioms (subsection \ref{sec:compaxioms}), our main results will be presented in Section \ref{sec:results}. Finally, the concluding remarks are given in Section \ref{section:conclusion}.

\subsection*{Related literature}

Consistency is a classic notion in social choice theory; it is popular in the study of various voting rules such as Borda rule \cite{Young1974a}, symmetric scoring rules \cite{Heckelman2020}, Approval Voting \cite{Fishburn1978,Vorsatz2007,Sato2014c}, Disapproval voting \cite{Alcantud2014}, Weighted approval voting \cite{Masso2008}. Our study introduces a variants of consistencies toward the study of SRSs.

The theory of SRS has grown rapidly in the last decade. From pioneering work \cite{Moretti2017}, a number of plausible solutions have been provided and axiomatically characterized: lex-cel \cite[]{Bernardi2019} and its variants \cite[]{Algaba2021,bengel2025extension,beal2022lexicographic,serramia2020qualitative,leroy2024ranking}, CP majority \cite[]{Haret2019a,fayard2018ordinal,Suzuki2026Unified}, ordinal Banzhaf index \cite[]{Khani2019,banzhaf1964weighted}, plurality \cite[]{Suzuki2025_Mill_SRS}, and intersection initial segment (IIS) \cite[]{Suzuki2026Ranking}. Some previous studies have also conducted a comparative study of multiple SRSs from a particular axiomatic viewpoint ---such as the ones using well-studied properties like desirability \cite[]{Aleandri2024a, maschler1966characterization, I58}  and independence \cite[]{Suzuki2025independence}--- or a game-theoretic perspective --- such us the articles exploring sabotage-proofness \cite[]{Suzuki2024c} or core-stability \cite[]{beal2023core}---. Focusing on the notion of consistency introduced in \cite{Suzuki2024}, the present study considers some new variants of consistency, and analyze lex-cel, sign lex-cel (a new SRS), and plurality in a comparative manner.

\section{Social ranking solutions}
\label{section:social ranking solutions}
For set $A$, we denote the set of all weak orders\footnote{A binary relation on $A$ that is reflexive, complete, and transitive.} on $A$ as $\mathcal{R}(A)$, and the set of all reflexive and complete binary relations on $A$ as $\mathcal{B}(A)$. 

Let $X$ be the finite set of individuals with $3\leq|X|<\infty$. Throughout the paper, $\vartriangleright$ denotes a fixed linear order on $X$. Let $\mathfrak{X}:=\mathfrak{P}(X)\setminus\{\emptyset\}$ be the set of all (logically possible) \textit{coalitions}. 

\begin{definition}
A \textit{social ranking solution} (\textit{SRS})  is a function $R:\domain\rightarrow\codomain$. 
\end{definition}

Our definition of the SRS follows \cite{Suzuki2024}. It maps a weak order $\succsim\in\domain$ to a reflexive and complete binary relation on $X$. For $\succsim\in\domain$, $R(\succsim)$ is often denoted as $R_{\succsim}$ for simplicity. The symmetric and asymmetric parts are denoted as $I_{\succsim}$ and $P_{\succsim}$ respectively (i.e., $xI_{\succsim}y\iff (xR_{\succsim}y\;\&\;yR_{\succsim}x)$ and $xP_{\succsim}y\iff (xR_{\succsim}y\;\&\;\neg yR_{\succsim}x)$). When the input $\succsim$ is a weak order on $\mathcal{D}\subseteq \mathfrak{X}$ (i.e., $\succsim\in\mathcal{R}(\mathcal{D})$), we call $\mathcal{D}$ as the \textit{coalition domain} of $\succsim$, denoted as $\mathcal{D}(\succsim)$---it is interpreted as the set of all comparable coalitions at $\succsim$. According to this terminology, an SRS maps a weak order on feasible coalitions to a reflexive and complete binary relation on the individuals. 

For $x\in X$ and $\mathcal{D}\subseteq \mathfrak{X}$, we denote the set of all elements (coalitions) in $\mathcal{D}$ that contain $x$ by $\mathcal{D}[x] := \{S \in \mathcal{D} \mid x \in S\}$. We also denote the intersection and union of all elements in $\mathcal{D}$ as $\bigcap \mathcal{D}:= \{x \in X \mid x \in S, \forall S  \in \mathcal{D}\}$ and $\bigcup \mathcal{D}:= \bigcup_{S\in\mathcal{D}}S$. For a nonnegative integer $l$, we denote by $[l]$ the set of all positive integers equal to or smaller than $l$, that is, $[l]:=\{b\in\mathbb{Z}\mid 1\leq b\leq l\}$.

The purpose of this section is to introduce several SRSs for the axiomatic analysis (Definition \ref{definition:lex-cel and S lex-cel}, \ref{definition:pluralities}, \ref{definition:CP}, and \ref{definition:IIS}). For these definitions, we first need to introduce some extra notation.  

A weak order $\succsim\in\domain$ is often denoted by arraying its equivalence classes as follows $\succsim:\Sigma_{1}\succ\Sigma_{2}\succ\cdots\succ\Sigma_{l}$. Each $\Sigma_{k}$ is an equivalence class of $\succsim$, that is, for any $S\in\Sigma_{k}$, we have $\Sigma_{k}=\{T\in\mathcal{D}(\succsim)\mid S\sim T\}$; for any $k,k'\in[l]$ with $k<k'$, we assume that $S\succ T$ for all $S\in\Sigma_{k}$ and $T\in\Sigma_{k'}$. For $x\in X$ let $\Sigma_{k}[x]=\{S\in\Sigma_k\mid x\in S\}$ and define 
\begin{align*}
    &\theta_{\succsim}(x):=(x_{1},\cdots,x_{l}),\text{ where } x_{k}:=\lvert\Sigma_{k} [x]\rvert,\ \forall k\in[l];\\
    & \dot{\theta}_{\succsim}(x):=(\dot{x}_{1},\cdots,\dot{x}_{l}), \text{ where }\dot{x}_{k}:=\He(x_{k})(=\He(\lvert\Sigma_{k}[x]\rvert)), 
\end{align*}
where, for all $k=1, \ldots, l$, 
 $\He(x_{k})$ is defined as follows: 
$$\He(x_{k}):=\begin{cases}
1&\text{ if }x_{k}>0,\\
0&\text{ if }x_{k}\leq0.\\
\end{cases}
$$

\begin{definition}
\label{definition:lex-cel and S lex-cel}
 For any $\succsim\in\domain$ and $x,y\in X$, we define \textit{lexicographic excellence solution} (\textit{lex-cel}), denoted as $R^{L}$, and \textit{Sign lexicographic excellence solution} (\textit{S lex-cel}), denoted as $R^{SL}$, as follows: 
\begin{align*}
&xR_{\succsim}^{L}y\iff \theta_{\succsim}(x)\geq^{L}\theta_{\succsim}(y),\\ 
&xR^{SL}_{\succsim}y\iff \dot{\theta}_{\succsim}(x)\geq^{L}\dot{\theta}_{\succsim}(y),
\end{align*}
where, for $a,b\in\mathbb{R}^n$, $a\geq^{L}b$ whether $a=b$ or if $\exists k\in[n]$ such that $a_i=b_i$ for all $i<k$ and $a_k >b_k$.   
\end{definition}

Other solutions based on the cardinality of some specific equivalence classes are defined as follows.

\begin{definition}
\label{definition:pluralities}
For any $\succsim\in\domain$ with $\succsim:\Sigma_{1}\succ\cdots\succ\Sigma_{l}$ and $x,y\in X$, we define \textit{plurality} $R^{P}$, \textit{sign plurality} (\textit{S plurality}) $R^{SP}$, and \textit{anti-plurality} $R^{AP}$ as follows:
\begin{align*}
xR^{P}_{\succsim}y&\iff x_{1}\geq y_{1};\\
xR^{SP}_{\succsim}y&\iff \dot{x}_{1}\geq \dot{y}_{1};\\
xR^{AP}_{\succsim}y&\iff x_{l}\leq y_{l}.
\end{align*}
\end{definition}

For any $\succsim\in\domain$ and $x,y\in X$, we denote as $C_{xy}=|\{S \in \mathfrak{X}: S \cap \{x,y\}=\emptyset, S \cup \{x\}, S \cup \{y\} \in \mathcal{D}(\succsim), S \cup \{x\} \succ S \cup \{y\}  \}|$ the number of \textit{CP comparisons} \cite{Haret2019a} in favor of $x$ against $y$.
\begin{definition}
\label{definition:CP}
 For any $\succsim\in\domain$ and $x,y\in X$.
 We define \textit{Ceteris Paribus majority solution} (\textit{CP majority}), denoted as $R^{CPM}$, as follows: 
\begin{align*}
&xR_{\succsim}^{CPM}y\iff C_{xy}\geq C_{yx}. 
\end{align*}

\end{definition}

 For any $\succsim\in\domain$ such that $\succsim:\Sigma_1\succ\ldots\succ\Sigma_l$, let $T_{k} := \bigcap  ( \bigcup_{j \leq k} \Sigma_{j} ) = \{ x \in X: \forall j \leq k,\forall S \in \Sigma_{j},x \in S \}$. For every $x \in X$, let $e_{\succsim}(x)$ be an integer such that
\[
e_{\succsim}(x):=
\begin{cases}
\max\{k\in [l]\mid x\in T_{k}\} & \text{if } x\in T_{1},\\
0                               & \text{otherwise.}
\end{cases}
\]
\begin{definition}
\label{definition:IIS}
For any $\succsim\in\domain$ and $x,y\in X$, we define \textit{intersection initial segment} (\textit{IIS}) $R^{IIS}$ as follows:
\begin{align*}
R^{IIS}_{\succsim}y \iff e_{\succsim}(x)\geq e_{\succsim}(y).
\end{align*}
\end{definition}

\begin{example}
Let $X=\{x,y,z,w\}$, $\succsim_{1}:\{\{x,y,z\},\{x,y\}\}\succ_{1}\{\{x\},$ $\{x,z\},\{y\}\}\succ_{1}\{\{z\},\{w\}\}$ and $\succsim_{2}:\{\{x,w\},\{y,w\},\{z,w\}\}$. Then, we obtain the following table.
\[
\begin{array}{c|c|c|c|c|c|c}
 & \theta_{\succsim_{1}}(\cdot)
 & \dot{\theta}_{\succsim_{1}}(\cdot)
 & e_{\succsim_{1}}(\cdot)
 & \theta_{\succsim_{2}}(\cdot)
 & \dot{\theta}_{\succsim_{2}}(\cdot)
 & e_{\succsim_{2}}(\cdot) \\[3pt]
\hline
x & (2,2,0) & (1,1,0) & 1 & (1) & (1) & 0 \\
y & (2,1,0) & (1,1,0) & 1 & (1) & (1) & 0 \\
z & (1,1,1) & (1,1,1) & 0 & (1) & (1) & 0 \\
w & (0,0,1) & (0,0,1) & 0 & (3) & (1) & 1
\end{array}
\]
Accordingly, we can verify the following: 
\begin{align*}
   &R^{L}_{\succsim_{1}}:x\succ y\succ z\succ w,\,  R^{SL}_{\succsim_{1}}:z\succ\{x,y\}\succ \{w\},\, R^{SP}_{\succsim_{1}}:\{x,y,z\}\succ w;\\
   &R^{CPM}_{\succsim_{1}}=\{(x,y),(y,x),(x,z),(x,w),(y,z),(y,w),(z,w),(w,z)\};\\
   &R^{P}_{\succsim_{1}}:\{x,y\}\succ z\succ w,\quad R^{AP}_{\succsim_{1}}:\{x,y\}\succ\{z,w\},\quad R^{IIS}_{\succsim_{1}}=R^{AP}_{\succsim_{1}};\\
   &R^{L}_{\succsim_{2}}:w\succ\{x,y,z\},\quad R^{SL}_{\succsim_{2}}:\{x,y,z,w\},\quad R^{CPM}_{\succsim_{2}}=X\times X;\\
   &R^{P}_{\succsim_{2}}: w\succ\{x,y,z\},\quad R^{SP}_{\succsim_{2}}: X,\quad R^{AP}_{\succsim_{2}}:\{x,y,z\}\succ w,\  R^{IIS}_{\succsim_{2}} = R^{P}_{\succsim_{2}}.
\end{align*} 
\end{example}

Let us briefly explain the basic ideas of each SRS. In comparing individuals $x$ and $y$ at given coalitional ranking $\succsim:\Sigma_{1}\succ\Sigma_{2}\succ\cdots$, lex-cel compares the frequencies of $x$ and $y$ lexicographically from the best equivalence class; S lex-cel sees only whether the individuals appear (1) or do not appear (0) in each equivalence class, comparing the indexes lexicographically; plurality (anti-plurality) compares only the frequencies at the top (bottom) equivalence class; IIS focuses on the intersection of each equivalence class---it declares that $x$ is better than $y$ if there exists $\hat{k}$ such that $x$ belongs to the intersections of $\Sigma_{1},\Sigma_{2},\cdots,\Sigma_{\hat{k}}$ simultaneously, but $y$ does not; and CP majority compares $x$ and $y$ ceteris paribus---comparing $x$'s coalitions and $y$'s coalitions with the other teammates fixed (i.e., $S\cup\{x\}$ versus $S\cup\{y\}$ for $S\subseteq X\setminus\{x,y\}$). 

Lex-cel \cite[]{Bernardi2019,Aleandri2024a,Suzuki2024,Suzuki2024c,Suzuki2025independence}, plurality \cite[]{Suzuki2025_Mill_SRS,Suzuki2025independence}, CP majority \cite[]{Haret2019a, Aleandri2024a} and IIS \cite[]{Suzuki2026Ranking,Suzuki2025independence} are found in the literature. S lex-cel (a variant of lex-cel) and S plurality (a variant of plurality) are new SRSs. Although these two SRSs---which are solely based on the presence, or absence, of individuals in the various equivalence classes, rather than their frequency---may appear rather rough, we would like to stress that in certain situations, equivalence classes may represent a benchmark for the evaluation. This is the case, for example, when ranking research teams based on their publication records: publishing at least one paper in a top journal could already be a sufficient criterion for achieving a high ranking that is deserving of a prize.

\section{Consistency}\label{axioms_consist}
For any $\succsim_{1},\succsim_{2}\in\domain$, we say that $\succsim_{1}$ and $\succsim_{2}$ are \textit{disjoint} if $\mathcal{D}\left(\succsim_{1}\right)\cap\mathcal{D}\left(\succsim_{2}\right)=\emptyset$. For any disjoint $\succsim_{1},\succsim_{2}\in\domain$, we define: 
\[\succsim_{1}+\succsim_{2}:=\{\succsim\in\mathcal{R}(\mathcal{D}(\succsim_{1})\cup\mathcal{D}(\succsim_{2})): \forall i\in\{1,2\},\succsim\mid_{\mathcal{D}(\succsim_{i})}=\succsim_{i}\}.\]
where $\succsim\mid_{\mathcal{D}(\succsim_{i})}$ denotes the restriction of $\succsim$ to the elements of $\mathcal{D}(\succsim_{i})$; i.e., $\succsim\mid_{\mathcal{D}(\succsim_{i})}\in \mathcal{R}(\mathcal{D}(\succsim_{i}))$ and $S \succsim\mid_{\mathcal{D}(\succsim_{i})} T \Leftrightarrow S \succsim_{i} T$ for all $S,T \in \mathcal{D}(\succsim_{i})$.
An element of $\succsim_{1}+\succsim_{2}$ is called a \textit{sum} of $\succsim_{1}$ and $\succsim_{2}$. Some sums are introduced as follows. 

\begin{itemize}
\item For any disjoint $\succsim_{1},\succsim_{2}\in\domain$, \textit{the concatenation sum} of $\succsim_{1}$ and $\succsim_{2}$, denoted as $\succsim=\succsim_{1}\cdot\succsim_{2}$, is the sum of $\succsim_{1}$ and $\succsim_{2}$, such that $S\succ T$ holds for all $S\in\mathcal{D}(\succsim_{1})$ and $T\in\mathcal{D}(\succsim_{2})$. 
\item For any disjoint $\succsim_{1},\succsim_{2}\in\domain$ with $\succsim_{1}:\Sigma_{1}\succ_{1}\cdots\succ_{1}\Sigma_{l}$ and $\succsim_{2}:\Gamma_{1}\succ_{2}\cdots\succ_{2}\Gamma_{m}$, \textit{the top-aligned sum} of $\succsim_{1}$ and $\succsim_{2}$, denoted as $\succsim_{1}\topsum\succsim_{2}$, is the sum of $\succsim_{1}$ and $\succsim_{2}$, such that $\Sigma_{k}$ and $\Gamma_{k}$ belong to the same equivalence class for each $k\in[\min\{l,m\}]$.
\item For any disjoint $\succsim_{1},\succsim_{2}\in\domain$ with $\succsim_{1}:\Sigma_{1}\succ_{1}\cdots\succ_{1}\Sigma_{l}$ and $\succsim_{2}:\Gamma_{1}\succ_{2}\cdots\succ_{2}\Gamma_{m}$, \textit{the bottom-aligned sum} of $\succsim_{1}$ and $\succsim_{2}$, denoted as $\succsim_{1}\bottomsum\succsim_{2}$, is the sum of $\succsim_{1}$ and $\succsim_{2}$, such that $\Sigma_{l+1-k}$ and $\Gamma_{m+1-k}$ belong to the same equivalence class for each $k\in[\min\{l,m\}]$.
\end{itemize}

It is straightforward in that for any disjoint $\succsim_{1},\succsim_{2}\in\domain$, the concatenation sum of $\succsim_{1}$ and $\succsim_{2}$, the top-aligned sum of $\succsim_{1}$ and $\succsim_{2}$, and the bottom-aligned sum of $\succsim_{1}$ and $\succsim_{2}$ all uniquely exist. Because all of them are sums of $\succsim_{1}$ and $\succsim_{2}$, they order the elements of $\mathcal{D}(\succsim_{i})$ according to $\succsim_{i}$ ($i=1,2$). Their difference is the relationship between the elements in $\mathcal{D}(\succsim_{1})$ and the elements in $\mathcal{D}(\succsim_{2})$. Under the concatenation sum $\succsim_{1}\cdot\succsim_{2}$, every element of $\mathcal{D}(\succsim_{1})$ is ranked above $\mathcal{D}(\succsim_{2})$; under the top-aligned sum $\succsim_{1}\topsum\succsim_{2}$, for each $i=1,2$ and $k=1,2,\cdots,\min\{l, m\}$, elements in the $k^{\text{th}}$ equivalence class in $\succsim_{i}$ remains in the $k^{\text{th}}$ equivalence class in $\succsim_{1}\topsum\succsim_{2}$, too (in other words, $\succsim_{1}$ and $\succsim_{2}$ are unified by aligning the positions of the top equivalence classes); under the bottom-aligned sum $\succsim_{1}\bottomsum\succsim_{2}$, the positions of the bottom equivalence classes are aligned (Example \ref{example:sums}). 
\begin{example}
\label{example:sums}
Let $\succsim_{1}:\Sigma_{1}\succ_{1}\Sigma_{2}\succ_{1}\Sigma_{3}$ and $\succsim_{2}:\Gamma_{1}\succ_{2}\Gamma_{2}$. Then, the concatenation sum  $\succsim_{3}:=\succsim_{1}\cdot\succsim_{2}$, top-aligned sum $\succsim_{4}:\succsim_{1}\topsum\succsim_{2}$, and the bottom-aligned sum $\succsim_{5}:\succsim_{1}\bottomsum\succsim_{2}$ are as follows: 
\begin{align*}
&\succsim_{3}:\Sigma_{1}\succ_{3}\Sigma_{2}\succ_{3}\Sigma_{3}\succ_{3}\Gamma_{1}\succ_{3}\Gamma_{2},\\
&\succsim_{4}:\Sigma_{1}\cup\Gamma_{1}\succ_{4}\Sigma_{2}\cup\Gamma_{2}\succ_{4}\Sigma_{3},\text{ and}\\
&\succsim_{5}:\Sigma_{1}\succ_{5}\Sigma_{2}\cup\Gamma_{1}\succ_{5}\Sigma_{3}\cup\Gamma_{2}.
\end{align*}
\end{example}

\begin{definition}
\label{definition:consistencies}
An SRS $R:\domain\rightarrow\codomain$ satisfies 
\begin{itemize}
\item \textit{consistency} (CON), 

if for any $\succsim_{1},\succsim_{2}\in\domain$, and $\succsim\in\succsim_{1}+\succsim_{2}$, and $x,y\in X$, the following relations (1)-(4) hold,
\item  \textit{concatenation consistency} (CCON), 

if for any $\succsim_{1},\succsim_{2}\in\domain$ and $\succsim:=\succsim_{1}\cdot\succsim_{2}$ and for any $x,y\in X$, the following relations (1)-(4) hold,
\item \textit{top-aligned consistency} (TCON),

if for any $\succsim_{1},\succsim_{2}\in\domain$ and $\succsim:=\succsim_{1}\topsum\succsim_{2}$ and for any $x,y\in X$, the following relations (1)-(4) hold,
\item \textit{bottom-aligned consistency} (BCON),

if for any $\succsim_{1},\succsim_{2}\in\domain$ and $\succsim:=\succsim_{1}\bottomsum\succsim_{2}$ and for any $x,y\in X$, the following relations (1)-(4) hold, 
\end{itemize}
where 
\begin{enumerate}[label=(\arabic*)]
\item if $xI_{\succsim_1}y$ and $xI_{\succsim_2}y$, then $xI_{\succsim}y$,
\item if $xI_{\succsim_1}y$ and $xP_{\succsim_2}y$, then $xP_{\succsim}y$,
\item if $xP_{\succsim_1}y$ and $xI_{\succsim_2}y$, then $xP_{\succsim}y$,
\item if $xP_{\succsim_1}y$ and $xP_{\succsim_2}y$, then $xP_{\succsim}y$.
\end{enumerate}
\end{definition}

We say that an SRS $R$ satisfies \textit{II-concatenation consistency} (\textit{II-CCON}) if $R$ satisfies the condition (1) in Definition \ref{definition:consistencies} for all disjoint $\succsim_{1},\succsim_{2}\in\domain$ with $\succsim:=\succsim_{1}\cdot\succsim_{2}$ and $x,y\in X$. Similarly, IP-concatenation consistency (IP-CCON), PI-concatenation consistency (PI-CCON), and PP-concatenation consistency (PP-CCON) are defined by conditions (2)-(4), respectively. Clearly, CCON is equivalent to the combination of II-CCON, IP-CCON, PI-CCON, and PP-CCON.

\subsection{Relations among consistency axioms}
\label{subsection:relations among consistency axioms}
The logical relationship between the three main consistencies (CON, CCON, TCON, and BCON) is summarized in Proposition \ref{proposition:relationship of consistencies} (and visualized in Fig.\ref{figure:consistencies}). 

\begin{proposition}
\label{proposition:relationship of consistencies}
Each of the following holds.  
\begin{enumerate}
\item CON implies the other three---that is, if an SRS $R$ satisfies CON, then it also satisfies CCON, TCON, and BCON.
\item Among the three consistencies axioms CCON, TCON, and BCON, S lex-cel satisfies only CCON; plurality satisfies only TCON; anti-plurality satisfies only BCON; IIS and CP majority satisfy none of them. 
\end{enumerate}
\end{proposition}

\begin{proof}
(1) is straightforward by the definition of the axioms. We will prove (2). Let $x,y,z\in X$ be distinct elements. On \textbf{S lex-cel}: it clearly satisfies CCON. Let $\succsim_{1}:\{\{x,y\}\}$ and $\succsim_{2}:\{\{x\}\}$. Then, we have $xI^{SL}_{\succsim_{1}}y$, $xP^{SL}_{\succsim_{2}}y$, and $xI^{SL}_{\succsim}y$ for $\succsim$ equals to $\succsim_{1}\topsum\succsim_{2}$ and $\succsim_{1}\bottomsum\succsim_{2}$, respectively. This contradicts TCON and BCON. On \textbf{plurality}: it clearly satisfies TCON. Let $\succsim_{1}:\{\{x,y\}\}\succ_{1}\{\{z\}\}$ and $\succsim_{2}:\{\{x\}\}$. Then, we have $xI^{P}_{\succsim_{1}}y$, $xP^{P}_{\succsim_{2}}y$, and $xI^{P}_{\succsim}y$ for $\succsim$ equals to $={\succsim_{1}\cdot\succsim_{2}}$ and $\succsim_{1}\bottomsum\succsim_{2}$, respectively. So, plurality violates CCON and BCON. \textbf{\textbf{Anti-plurality}} can also be verified in a dual way.  On \textbf{IIS}: let $\succsim_{1}:\{\{z\}\}$ and $\succsim_{2}:\{\{x\}\}$. Then, we have $xI^{IIS}_{\succsim_{1}}y$, $xP^{IIS}_{\succsim_{2}}y$, and $xI^{IIS}_{\succsim}y$ for $\succsim$ equals to $\succsim_{1}\cdot\succsim_{2}$, $\succsim_{1}\topsum\succsim_{2}$, and $\succsim_{1}\bottomsum\succsim_{2}$, respectively. Thus, IIS violates all three consistencies. On \textbf{CP majority}: let $\succsim_{1}:\{\{x\}\}\succ_{1}\{\{x,y,z\}\}$ and $\succsim_{2}:\{\{z\}\}\succ_{2}\{\{y\}\}$. Then, we find that $xI^{CPM}_{\succsim_{1}}y$, $xI^{CPM}_{\succsim_{2}}y$, and $xP^{CPM}_{\succsim}y$ for $\succsim$ are equal to  $\succsim_{1}\cdot\succsim_{2}$, $\succsim_{1}\topsum\succsim_{2}$, and $\succsim_{1}\bottomsum\succsim_{2}$, respectively. This proves that the CP majority fails to satisfy any one of the three consistencies.
\end{proof}

\begin{figure}[ht]
\centering
\includegraphics[width=0.6\linewidth]{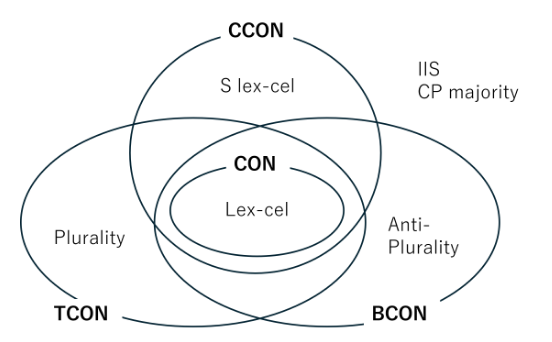}
\caption{Consistencies and SRSs}
\label{figure:consistencies}
\end{figure}


    
    
    

Figure \ref{figure:consistencies} visualizes the result stated in Proposition \ref{proposition:relationship of consistencies}. It is worth noting that each lex-cel, S lex-cel, plurality, anti-plurality, and IIS (or CP majority) is located in the different areas in the Venn Diagram. This indicates that the three consistency axioms (CCON, TCON, and BCON), which are introduced in the present paper, successfully grasp the different characteristics of the familiar SRSs. 

It is also noteworthy that many of the existing SRSs---plurality, anti-plurality, IIS, and CP majority---do not satisfy CCON; however, this is due to different reasons. To see this, let us further look at which of (1)-(4) in Definition \ref{definition:consistencies} are violated by these SRSs. 

\begin{proposition}
\label{proposition:concatenation consistencies}
Each of the following holds. 
\begin{enumerate}
\item IIS and plurality satisfy II-CCON, PI-CCON, and PP-CCON (but not IP-CCON). 
\item Dual IIS and anti-plurality satisfy II-CCON, IP-CCON, and PP-CCON (but not PI-CCON). 
\item If $\lvert X\rvert\geq 4$, CP majority satisfies none of II-CCON, IP-CCON, PI-CCON, and PP-CCON.  
\end{enumerate}
\end{proposition}

\begin{proof}
(1) and (2) are straightforward (and partly overlap with the proof of Proposition \ref{proposition:relationship of consistencies}). Hence, we will only prove (3). II-CCON is done in the proof of Proposition \ref{proposition:relationship of consistencies}. Let $x,y,z,w\in X$ be distinct elements. \textbf{IP-CCON}: let $\succsim_{1}:\{\{z\}\}$ and $\succsim_{2}:\{\{x,z\},\{x\}\}\succ_{2}\{\{y,z\}\}$. Then, we can observe that $xI^{CPM}_{\succsim_{1}}y$, $xP^{CPM}_{\succsim_{2}}y$, and $xI^{CPM}_{\succsim_{1}\cdot\succsim_{2}}y$. So, CP majority violates IP-CCON. \textbf{PI-CCON} is also violated for similar reasons. \textbf{PP-CCON}: let $\succsim_{1}:\{\{x\}\}\succ_{1}\{\{y\},\{y,w\},\{y,z,w\},$ $\{y,z\}\}$ and $\succsim_{2}:\{\{x,z\}\}\succ_{2}\{\{y,z\},\{x,w\},\{x,z,w\}\}$. Then, we have $xP^{CPM}_{\succsim_{1}}y$, $xP^{CPM}_{\succsim_{2}}y$, and $xI^{CPM}_{\succsim_{1}\cdot\succsim_{2}}y$. This contradicts PP-CCON. 
\end{proof}

Figure \ref{figure:concatenation consistencies} summarizes Proposition \ref{proposition:concatenation consistencies}. While the previous Figure \ref{figure:consistencies} demonstrates that many of the existing SRSs (plurality, CP majority, IIS, etc.) fail to meet the full demand of CCON, Figure \ref{figure:concatenation consistencies} shows the extent to which each SRS satisfies/fails to satisfy the CCON. An important fact from this result is that II-CCON and PP-CCON are satisfied by all IIS, plurality, S lex-cel, lex-cel, and anti-plurality. In other words, if the social ranking between $x$ and $y$ are the same in two disjoint coalitional rankings $\succsim_{1}$ and $\succsim_{2}$ (i.e., $x$ and $y$ are judged indifferent in both $\succsim_{1}$ and $\succsim_{2}$, or $x$[$y$] is judged better than $y$[$x$] in both $\succsim_{1}$ and $\succsim_{2}$), then it must also be in the concatenated sum $\succsim_{1}\cdot\succsim_{2}$. Therefore, non-consistent results (violation of CCON) can occur only when the social ranking between two individuals $x$ and $y$ differ between the two coalitional rankings $\succsim_{1}$ and $\succsim_{2}$. 

\begin{figure}[ht]
\centering
\includegraphics[width=0.5\linewidth]{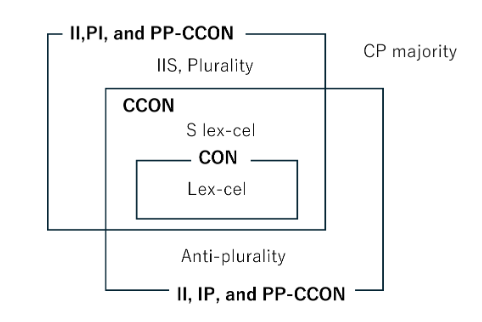}
\caption{Concatenation consistencies and SRSs}
\label{figure:concatenation consistencies}
\end{figure}


    
    
    




\section{Complementary Axioms}\label{sec:compaxioms}

In this section, we introduce some complementary axioms that will be useful later to axiomatically characterize the SRSs introduced in the previous sections using the different notions of consistency. The first two axioms address notions of invariance, either with respect to permutations of players or to permutations of coalitions.
Before formally introducing these axioms, we need to first introduce some preliminary notation. 
A bijection $\sigma:X\rightarrow X$ is said to be a permutation $\sigma$ on $X$. For any permutation $\sigma$ on $X$ and any relation $\succsim \in \domain$, we denote as $\succsim^\sigma \in \domain$ the relation such that $S \succsim T \Leftrightarrow \sigma(S)  \succsim^\sigma \sigma(T)$, where $\sigma(S)=\{\sigma(x): x\in S\}$.
A bijection $\pi:\mathfrak{X}\rightarrow \mathfrak{X}$ is said to be a permutation $\pi$ on $\mathfrak{X}$. For any $x \in X$, a permutation $\pi$ on $\mathfrak{X}$ is seen as $x$-\emph{invariant} if for any $S \in \mathfrak{X}$ with $x \in S$ we have that $x \in \pi(S)$; it is said $\{x,y\}$-\emph{invariant} if it is both $x$-invariant and $y$-invariant, for any $x,y \in X$. For any permutation $\pi$ on $\mathfrak{X}$ and any relation $\succsim \in \domain$, we denote as $\succsim_\pi \in \domain$ the relation such that $S \succsim T \Leftrightarrow \pi(S) \succsim_\pi \pi(T)$, for all $S,T\in\mathfrak{X}$.

\begin{definition}
\label{definition:NT and WCA}
An SRS $R:\domain\rightarrow\codomain$ satisfies
\begin{itemize}
\item \textit{neutrality} (\textit{NT}) if, for any $\succsim\in\domain$ and any permutation $\sigma$ on $X$, we have that $(R_{\succsim})^\sigma=R_{\succsim^{\sigma}}$, where $(R_{\succsim})^\sigma=\{(\sigma(x),\sigma(y))\mid (x,y)\in R_{\succsim})\}$.

\item \textit{weak coalitional anonymity} (\textit{WCA}) if, for any $x,y\in X$ and for any $\{x,y\}$-invariant permutation $\pi$ on $\mathfrak{X}$, we have $R_{\succsim}\mid_{\{x,y\}}=R_{\succsim_\pi
}\mid_{\{x,y\}}$ for all $\succsim\in\domain$.
\end{itemize}
\end{definition}

NT axiom is a very basic requirement that states that an SRS should not depend on the name of individuals in $X$ and has been widely used in the literature \cite{Bernardi2019, Algaba2021}. WCA axiom is also a classical property for SRS \cite{Algaba2021, Suzuki2024} and demands to rank individual parties solely according to the positions of the coalitions they belong to, regardless of the coalitions’ identities or sizes.
{A second family of axioms is aimed at preserving the independence of the social ranking across coalitional rankings that only differ for a modification of the position of certain coalitions. In particular, the following two mimic the definition of the {\it Independence of the Worst Set} (IWS) \cite{Bernardi2019,Algaba2021,Aleandri2024a,Suzuki2026Ranking}, that has been used to axiomatically characterize the lex-cel solution over a domain of complete coalitional rankings.

\begin{definition}
\label{definition:IDWS and IAWS}
An SRS $R:\domain\rightarrow\codomain$ satisfies:
\begin{enumerate}
\item \textit{independence of the decomposition of the worst sets} (\textit{IDWS}) if, for any $\succsim:\Sigma_{1}\succ\cdots\succ\Sigma_{l}$, where $l\geq 2$, and $\succsim':\Sigma_{1}\succ'\cdots\succ'\Sigma_{l-1}\succ'\Gamma_{1}\succ'\cdots\succ'\Gamma_{m}$ with $\Sigma_{l}=\Gamma_{1}\cup\ldots\cup\Gamma_{m}$ and $\Gamma_{i}\cap\Gamma_{j}=\emptyset$, for all $i,j\in[m]$, and, for any $x,y\in X$, we have that $\big[\ xP_{\succsim}y\Rightarrow xP_{\succsim'}y\ \big]$.
\item \textit{independence of the addition of the worst sets} (\textit{IAWS}) if for any $\succsim:\Sigma_{1}\succ\cdots\succ\Sigma_{l}$ and $\succsim':\Sigma_{1}\succ'\cdots\succ'\Sigma_{l}\succ'\Gamma$ with $\Gamma\subseteq\mathfrak{X}\setminus\mathcal{D}(\succsim)$, and for any $x,y\in X$, we have that $\big[\ xP_{\succsim}y\Rightarrow xP_{\succsim'}y\ \big]$.
\end{enumerate}
\end{definition}

Following the standard interpretation of IWS for complete coalitional rankings, IDWS can be interpreted as a principle designed to reward the excellence of coalitions when the domain is variable ($\codomain$). In fact, property states that if the positions of coalitions within the worst equivalence class of a coalitional ranking are modified, as long as they remain dominated by the other coalitions in $\mathcal{D}$, then the strict ranking of the individuals should remain unaffected.
However, it is noteworthy that in our model, where the set of all feasible coalitions can vary, lex-cel does not satisfy IDWS, as shown in the following proposition. 

\begin{proposition}
\label{proposition:lex-cel_IWS}
In the setting of a variable domain for coalitions, lex-cel does not satisfy IDWS.
\end{proposition}

\begin{proof}
Let 
\begin{align*}
&\succsim:\{\{3\}\}\succ\{\{1\},\{1,3\},\{2\}\},\\
&\succsim':\{\{3\}\}\succ'\{\{2\}\}\succ'\{\{1\},\{1,3\}\}.
\end{align*}
Then, we have $1P_{\succsim}^L2$ and $2P_{\succsim'}^L1$. This contradicts IDWS (as $\succsim'$ is obtained from $\succsim$ by decomposing the worst class). 
\end{proof}
Unlike IDWS, IAWS does not extend the IWS property to the variable domains; however, it still incorporates a similar principle, highlighting the irrelevance of newly formed coalitions occupying the worst positions in a coalitional ranking when determining a strict social relation.
It is straightforward to show that some familiar SRSs including lex-cel (such as S lex-cel $R^{SL}$, IIS $R^{IIS}$, and plurality $R^{P}$) all satisfy IAWS. 

An even stronger property rewarding individuals' position in excellent coalitions is the following property, stating that only the top equivalence class of coalitional ranking should be considered for social ranking.
\begin{definition}
An SRS $R:\domain\rightarrow\codomain$ satisfies \textit{tops-only} (TO) if, for any $\succsim,\succsim'\in\domain$ with $\succsim:\Sigma_{1}\succ\cdots\succ\Sigma_{l}$ and $\succsim':\Sigma_{1}'\succ'\cdots\succ'\Sigma_{l}'$, and $\Sigma_{1}=\Sigma_{1}'$, then, it follows that $R_{\succsim}=R_{\succsim'}$.
\end{definition}

Clearly, TO is logically stronger than IWAS and IWDS: an SRS that satisfies TO also satisfies IWAS and IWDS, as it only depends on top coalitions---no matter which coalitions belong to the other equivalence classes or how they are ranked.

We conclude this section with a third class of axioms that specify how to socially rank individuals on particular coalitional ranking. The first axiom in this family is the axiom named \textit{all indifference all winners}  in \cite{Konieczny2022}, which is a kind of ``unanimity'' property for SRSs specifying that if all coalitions are indifferent, then all individuals should be considered equally important. Here, we generalize this axiom to our framework with a variable domain of coalitions, so that only individuals appearing in at least one coalition in $\mathcal{D}(\succsim)$ are considered equally important.
\begin{definition}
\label{definition:all indifference all winners}
An SRS $R:\domain\rightarrow\codomain$ satisfies \textit{all indifference all winners} (\textit{AIAW}) if, for any $\succsim\in\domain$ such that $S\sim T$ for all $S,T\in \mathcal{D}(\succsim)$,
$x I_{\succsim} y$ and $x P_{\succsim} z$, for all $x,y\in\bigcup \mathcal{D}(\succsim)$ and $z\notin\bigcup \mathcal{D}(\succsim)$.
\end{definition}

The last complementary axiom of this section is inspired by the theory of simple games. In fact, it deals with situations where $\succsim\in\domain$ partitions a set of coalitions in two equivalence classes $\Sigma_{1}\succ\Sigma_{2}$, so that $\Sigma_{1}$ could represent winning coalitions and $\Sigma_{2}$ losing coalitions. The \textit{Weak union very important person property} says that individuals belonging to some winning coalitions should be ranked strictly better than individuals that never take part in any winning coalitions---this is in line with the what is also suggested by marginal-type power indices \cite{dubey1981value}. In the context of social ranking, a property with a similar interpretation was also studied in \cite{Suzuki2026Ranking}.}

\begin{definition}
An SRS $R:\domain\rightarrow\codomain$ satisfies \textit{Weak union very important person property} (\textit{WUVIP}) if, 
for any  ranking $\succsim\in\domain$ with $\succsim:\Sigma_{1}\succ\Sigma_{2}$, and, for any $x,y\in X$, such that $x\in\bigcup\Sigma_{1}$ and $y\notin\bigcup\Sigma_{1}$, then $xP_{\succsim}y$.
\end{definition}
There are logical dependencies between some of the axioms introduced so far, as shown in the following proposition.
\begin{proposition}
 Let $R$ be an SRS satisfying CCON, AIAW, and IDWS. Then, it also satisfies WUVIP.   
\end{proposition}
\begin{proof}\label{prop:depend}
Consider a ranking $\succsim\in\domain$ with $\succsim:\Sigma_{1}\succ\Sigma_{2}$, and $x,y\in X$, such that $x\in\bigcup\Sigma_{1}$ and $y\notin\bigcup\Sigma_{1}$.
First, we want to prove that $x P_\succsim y$.
Notice that it must exist $S \notin \Sigma_1$ with $\{x,y\} \subseteq S$.
We distinguish two cases: $S \in \Sigma_2$ and $S \notin \Sigma_2$.

Suppose first that $S \in \Sigma_2$.
Consider the rankings $\succsim_1:\Sigma_{1}$ and $\succsim_2:\Sigma_{2}$.
By AIAW, we have $x P_{\succsim_1} y$ and $x I_{\succsim_2} y$. Then, by CCON and the fact that $\succsim=\succsim_1\cdot \succsim_2$, we have that $x P_{\succsim} y$.

Suppose now that $S \notin \Sigma_2$. Consider the ranking $\succsim'_2:\Sigma_{2}\cup \{S\}$.
By AIAW, we have  $x I_{\succsim'_2} y$.
Consider the ranking $\succsim':\Sigma_{1}\succ' \Sigma_{2}\cup \{S\}$. By CCON and the fact that $\succsim'=\succsim_1\cdot \succsim'_2$, we have that $x P_{\succsim'} y$.

Consider the ranking $\succsim'': \Sigma_{1}\succ'' \Sigma_{2}\succ'' \{S\}$. By IDWS on $\succ'$, we have that $x P_{\succsim''} y$. Finally, take the ranking $\succsim_3:\{S\}$.
By AIAW, $x I_{\succsim_3} y$. Then, by CCON on $\succsim''=\succsim\cdot \succsim_3$ and the fact proved earlier that  $x P_{\succsim''} y$, it must be $x P_{\succsim} y$ (otherwise, if $y R_{\succsim} x$ by CCON we would have $y R_{\succsim''} x$, a contradiction). This concludes the proof.
\end{proof}

\section{Main results}\label{sec:results}

In this section, we present our main results illustrating the role of consistency axioms in the analysis of SRSs. We start with an axiomatic characterization of S lex-cel.

\begin{theorem}
\label{theorem:SL}
Let $R$ be an SRS satisfying NT, CCON, and AIAW. Then, each of the following is logically equivalent: 
\begin{enumerate}
\item $R$ satisfies IAWS;
\item $R$ satisfies 
IDWS;
\item $R=R^{SL}$.
\end{enumerate}
\end{theorem}

\begin{proof}
Let $R$ be an SRS satisfying the stated axioms NT, CCON, and AIAW.

\noindent
\textit{\textbf{Proof of 1$\Rightarrow$3. }}
 Suppose that $R$ satisfies IAWS. Let $\succsim\in\domain$ and $x,y\in X$. We have to prove that $R=R^{SL}$. Therefore, we will prove that (i) if $\dot{\theta}_{\succsim}(x)=\dot{\theta}_{\succsim}(y)$, then $xI_{\succsim}y$; and (ii) if $\dot{\theta}_{\succsim}(x)>^{L}\dot{\theta}_{\succsim}(y)$, then $xP_{\succsim}y$. For each $k\in[l]$, we define the ranking $\succsim^{k}$ by $\succsim^{k}:\Sigma_{k}$.

Proof of (i): If $\dot{x}_{k}=\dot{y}_{k}=1$, then $xI_{\succsim^{k}}y$ by AIAW. If $\dot{x}_{k}=\dot{y}_{k}=0$, then $xI_{\succsim^{k}}y$ by NT. Therefore, we can conclude that $xI_{\succsim^{k}}y$  for all $k\in[l]$. By applying CCON repeatedly, we can infer that $xI_{\succsim}y$.

Proof of (ii): assume that $\dot{\theta}_{\succsim}(x)>^{L}\dot{\theta}_{\succsim}(y)$. By definition, there exists $\hat{k}\in[l]$, such that $x_{k}=y_{k}$ for all $k<\hat{k}$, and $\dot{x}_{\hat{k}}>\dot{y}_{\hat{k}}$. It follows that $\dot{x}_{\hat{k}}=1$ and $\dot{y}_{\hat{k}}=0$. Now, let
\begin{equation*}
\succsim_{1}:\Sigma_{1}\succ_{1}\cdots\succ_{1}\Sigma_{\hat{k}-1}\text{, and }
\succsim_{2}:\Sigma_{\hat{k}+1}\succ_{2}\cdots\succ_{2}\Sigma_{l}.
\end{equation*}
Then, we have $xI_{\succsim_{1}}y$ by case (i). By AIAW, we obtain $xP_{\succsim^{\hat{k}}}y$. Hence, CCON implies that $xP_{\succsim_{1}\cdot \succsim^{\hat{k}}}y$. By IAWS, we can say that $xP_{\succsim}y$, because $\succsim=(\succsim_{1}\cdot\succsim^{\hat{k}})\cdot\succsim_{2}$. 

\noindent
\textit{\textbf{Proof of 3$\Rightarrow$1.}}
Suppose that $R=R^{SL}$. Take $\succsim\in\domain$ with $\succsim:\Sigma_{1}\succ\cdots\succ\Sigma_{l}$ and $x,y\in X$ s.t. $x P^{SL}_{\succsim} y$. Take $\Gamma\subseteq\mathfrak{X}\setminus\mathcal{D}(\succsim)$ call $\succsim_{\Gamma}:\Gamma$ and define $\succsim'=\succsim \cdot\succsim_{\Gamma}$. By definition of $\dot{\theta}_{\succsim}(x)$ and $\dot{\theta}_{\succsim}(y)$, there exists $\hat{k}\in[l]$ such that $x_{k}=y_{k}$ for all $k<\hat{k}$, and $\dot{x}_{\hat{k}}>\dot{y}_{\hat{k}}$. By construction, we have that the first $\hat{k}$ components of $\dot{\theta}_{\succsim'}(x)$ and $\dot{\theta}_{\succsim'}(y)$ are equal to the components of $\dot{\theta}_{\succsim}(x)$ and $\dot{\theta}_{\succsim}(y)$, respectively. It follows that $\dot{\theta}_{\succsim'}(x)>^L\dot{\theta}_{\succsim'}(y)$ and axiom IAWS is satisfied.

\noindent
\textit{\textbf{Proof of 2$\Rightarrow$3}.}
Suppose that $R$ satisfies IDWS. By Proposition \ref{prop:depend} it also satisfies WUVIP. Let $\succsim\in\domain$ and $x,y\in X$. We will prove that (i) if $\dot{\theta}_{\succsim}(x)=\dot{\theta}_{\succsim}(y)$, then $xI_{\succsim}y$; (ii) if $\dot{\theta}_{\succsim}(x)>^{L}\dot{\theta}_{\succsim}(y)$, then $xP_{\succsim}y$. Case (i) can be shown as in the Proof of 1$\Rightarrow$3. So, we will prove (ii). By definition, there exists $\hat{k}\in[l]$ such that $x_{k}=y_{k}$ for all $k<\hat{k}$, and $\dot{x}_{\hat{k}}>\dot{y}_{\hat{k}}$. It follows that $\dot{x}_{\hat{k}}=1$ and $\dot{y}_{\hat{k}}=0$.
Let 
\begin{align*}
\succsim_{1}&:\Sigma_{1}\succ_{1}\cdots\succ_{1}\Sigma_{\hat{k}-1},\\
\succsim_{2}&:\Sigma_{\hat{k}}\succ_{2}\Sigma_{\hat{k}+1}\cup\cdots\cup\Sigma_{l},\\
\succsim_{3}&:\Sigma_{\hat{k}}\succ_{3}\Sigma_{\hat{k}+1}\succ_{3}\cdots\succ_{3}\Sigma_{l}.
\end{align*}

Then, we have $xI_{\succsim_{1}}y$ by case (i); $xP_{\succsim_{2}}y$ by WUVIP; $xP_{\succsim_{3}}y$ by IDWS; $xP_{\succsim}y$ by CCON as $\succsim=\succsim_1 \cdot \succsim_3$.

\noindent
\textit{\textbf{Proof of 3$\Rightarrow $2. }}
Suppose that $R=R^{SL}$. Take $\succsim\in\domain$ and $x,y\in X$ s.t. $x P^{SL}_{\succsim} y$. Take $\succsim''\in\domain$ obtained by $\succsim$ decomposing the last equivalence class; then, as before,the first $\hat{k}$ components of $\dot{\theta}_{\succsim''}(x)$ and $\dot{\theta}_{\succsim''}(y)$ are equal to the components of $\dot{\theta}_{\succsim}(x)$ and $\dot{\theta}_{\succsim}(y)$, respectively. It follows that $\dot{\theta}_{\succsim''}(x)>^L\dot{\theta}_{\succsim''}(y)$ and axiom IDWS is satisfied.
\end{proof}
\begin{remark}
   It is trivial to prove that $R^{SL}$ also satisfies WCA.
\end{remark}
We now introduce an axiomatic characterization of Plurality.

\begin{theorem}
\label{theorem:plurality}
An SRS $R:\domain\rightarrow\codomain$ satisfies NT, WCA, WUVIP, TCON, and TO if and only if $R=R^{P}$.
\end{theorem}

\begin{proof}
The "if" part is straightforward. We will prove the "only if" part. Let $R$ be an SRS satsifying the axioms. Let $\succsim\in\domain$ with $\succsim:\Sigma_{1}\succ\cdots\succ\Sigma_{l}$ and $x,y\in X$. We will prove that (i) if $x_{1}=y_{1}$, then $xI_{\succsim}y$; and (ii) if $x_{1}>y_{1}$, then $xP_{\succsim}y$. 

Proof of (i). Assume that $x_{1}=y_{1}$. Let $\succsim':\Sigma_{1}$. Then, we have $xI_{\succsim}y$ by NT and WCA. By TO, we have that $R_{\succsim}=R_{\succsim'}$. Hence, we can conclude that $xI_{\succsim'}y$.

Proof of (ii): assume that $x_{1}>y_{1}$. Let $\Gamma\subseteq \Sigma_{1}[x]\setminus\Sigma_{1}[y]$ with $\lvert \Gamma\rvert=x_{1}-y_{1}$. As $x_{1}>y_{1}$, $\Gamma$ is nonempty. Now, let
\begin{equation*}
\succsim_{1}:\Gamma,\quad  \succsim_{2}:\Gamma\succ_{2}\{\{x,y\}\},\quad  \succsim_{3}:\Sigma_{1}\setminus\Gamma,\quad  \succsim_{4}:\Sigma_{1}.
\end{equation*}
Then, we have $R_{\succsim_{1}}=R_{\succsim_{2}}$ by TO, and $xP_{\succsim_{2}}y$ by WUVIP. Hence, we have that $xP_{\succsim_{1}}y$. Furthermore, we have $xI_{\succsim_{3}}y$ by case (i). As $\succsim_{4}=\succsim_{1}\topsum\succsim_{3}$, TCON implies that $xP_{\succsim_{4}}y$. By TO, we obtain that $xP_{\succsim}y$.
\end{proof}

We conclude this section with an axiomatic characterization of lex-cel.
\begin{theorem}
\label{theorem:L2}
An SRS $R:\domain\rightarrow\codomain$ satisfies NT, WCA, WUVIP, CCON, TCON, and IAWS if, and only if, $R=R^{L}$.
\end{theorem}

\begin{proof}
Let $\succsim\in\domain$ with $\succsim:\Sigma_{1}\succ\cdots\succ\Sigma_{l}$ and $x,y\in X$. We are supposed to prove that (i) if $\theta_{\succsim}(x)=\theta_{\succsim}(y)$, then $xI_{\succsim}y$; and (ii) if $\theta_{\succsim}(x)>^{L}\theta_{\succsim}(y)$, then $xP_{\succsim}y$. 

Proof of (i): Define the ranking $\succsim^{k}$ by $\succsim^{k}:\Sigma_{k}$. We have $xI_{\succsim_{k}}y$ by WCA and NT. Therefore, CCON implies that $xI_{\succsim}y$. 

Proof of (ii): Assume that $\theta_{\succsim}(x)>^{L}\theta_{\succsim}(y)$. Let $\Gamma\subseteq \Sigma_{k}[x]\setminus\Sigma_{k}[y]$ with $\lvert\Gamma\rvert=x_{k}-y_{k}$. Let 
\begin{equation*}
\begin{aligned}
&\succsim_{1}:\Sigma_{1}\succ_{1}\cdots\succ_{1}\Sigma_{\hat{k}-1},
&\succsim_{2}:\Gamma,\quad\quad
&\succsim_{3}:\{\{x,y\}\},\\  
&\succsim_{4}:\Gamma\succ_{4}\{\{x,y\}\}(=\succsim_{2}\topsum\succsim_{3}),
&\succsim_{5}:\Sigma_{\hat{k}}\setminus\Gamma,
&\succsim_{6}:\Sigma_{\hat{k}}(=\succsim_{2}\topsum\succsim_{5}),\\
&\succsim_{7}:\Sigma_{1}\succ\cdots\succ\Sigma_{\hat{k}}. & &
\end{aligned}
\end{equation*}
Then, we have that $xI_{\succsim_{1}}y$ and $xI_{\succsim_{3}}y$ by (i); $R_{\succsim_{2}}\mid_{\{x,y\}}=R_{\succsim_{4}}\mid_{\{x,y\}}$ by CCON ($\because$ $\succsim_{4}=\succsim_{2}\cdot\succsim_{3}$ and $xI_{\succsim_{3}}y$); $xP_{\succsim_{4}}y$ by WUVIP. Hence, we obtain that $xP_{\succsim_{2}}y$. Furthermore, we also have that $xI_{\succsim_{5}}y$ by (i); $xP_{\succsim_{6}}y$ by TCON; $xP_{\succsim_{7}}y$ by IAWS.
\end{proof}

\subsection{Independence of the axioms}

In this section, we study the logical independence of the axioms used to axiomatically charcaterize SRSs in Section \ref{sec:results}. We first introduce some notions and new definitions of SRSs that will be considered in the following propositions to prove axioms independence. In other words, when studying a set of axioms investigated in Section \ref{sec:results}, each SRS satisfies a combination of all axioms but one. For our purpose, the value of SRSs introduced in this section lies solely in proving the logical independence of axioms studied in Section \ref{sec:results}; therefore, we will not waste words on their interpretation, which in some cases is completely meaningless.

For $\theta_{\succsim}(x)=(x_{1},\cdots,x_{l})$ and $\theta_{\succsim}(y)=(y_{1},\cdots,y_{l})$, let us write $\theta_{\succsim}(x)\geq^{DL}\theta_{\succsim}(y)$ if $\theta_{\succsim}(x)=\theta_{\succsim}(y)$ or there exists $\hat{k}\in[l]$ such that $x_{k}=y_{k}$ for all $k>\hat{k}$, and $x_{\hat{k}}<y_{\hat{k}}$. 
We now define other SRSs. 

\begin{definition}
\label{definition:variants of lex-cel and S lex-cel}
Some variants of (S) lex-cels are introduced as follows. For any $\succsim\in\domain$ and $x,y\in X$, 
\begin{itemize}
\item \textit{dual S lex-cel} $R^{DSL}$: $xR^{DSL}_{\succsim}y\iff \dot{\theta}_{\succsim}(x)\geq^{DL}\dot{\theta}_{\succsim}(y)$;
\item \textit{inverse dual S lex-cel }$R^{ISDL}$: $xR^{ISDL}_{\succsim}y\iff yR^{DSL}_{\succsim}x$.
\item \textit{Split lex-cel $R^{splitL}$} as follows. For any $\succsim\in\domain$ and $x\in X$, let $\theta^{split}_{\succsim}(x):=(x^{split}_{1},\cdots,x^{split}_{l})$, where $x^{split}_{k}:=\sum_{S\in\Sigma_{k}[x]}\frac{1}{\lvert S\rvert}$. Then, $xR^{split}_{\succsim}y\iff \theta^{split}_{\succsim}(x)\geq^{L}\theta^{split}_{\succsim}(y)$.
\item \textit{Lex-cel with tie-breaking by $\vartriangleright$}, denoted as $R^{L,\vartriangleright}$, is as follows: 
\[xR^{L,\vartriangleright}_{\succsim}y\iff (xP^{L}_{\succsim}y\text{ or }(xI^{L}_{\succsim}y\text{ and }x\vartriangleright y)).\]
\end{itemize}
\end{definition}
We now introduce further variants of lex-cel.
\begin{definition}
\label{definition:further variants of lex-cel}

 We introduce further variants of S lex-cel.
 For any $\succsim\in\domain$ with $\succsim:\Sigma_{1}\succ\cdots\succ\Sigma_{l}$ and $x,y\in X$, 
\begin{itemize}    
\item \textit{S lex-cel with a particular attention to the non-existence}, $R^{SLNE}$, is defined as follows: (a) if $\dot{\theta}_{\succsim}(x)=\dot{\theta}_{\succsim}(y)$ and there exists $k\in[l]$ such that $\dot{x}_{k}=\dot{y}_{k}=0$, then $R^{SLNE}_{\succsim}\mid_{\{x,y\}}=$ $\vartriangleright\mid_{\{x,y\}}$; and (b) otherwise, we have $R^{SLNE}_{\succsim}\mid_{\{x,y\}}=R^{SL}_{\succsim}\mid_{\{x,y\}}$.
\item \textit{S lex-cel with a particular attention to the non-existence in higher classes}, $R^{SLNEH}$, is defined as follows: (a) if there exists $\hat{k}\in[l]$ such that $\dot{x}_{k}=\dot{y}_{k}=1$ for all $k<\hat{k}$ and $\dot{x}_{\hat{k}}=\dot{y}_{\hat{k}}=0$, then $R^{SLNEH}_{\succsim}\mid_{\{x,y\}}=$ $\vartriangleright\mid_{\{x,y\}}$; and (b) otherwise, then $R^{SLNEH}_{\succsim}\mid_{\{x,y\}}=R_{\succsim}^{SL}\mid_{\{x,y\}}$.
\item \textit{S lex-cel with a particular precedence on unanimity}, $R^{SLUN}$, is defined as follows: 
$xR^{SLUN}_{\succsim}y\iff \tilde{\theta}_{\succsim}(x)\geq^{L} \tilde{\theta}_{\succsim}(y),$
where for any $z\in X$, we denote $\tilde{\theta}_{\succsim}(z)=(\tilde{z}_{1},\cdots,\tilde{z}_{l})$, where $$\tilde{z}_{k}=\begin{cases}
2 &\text{if }z_{k}=\lvert\Sigma_{k}\rvert,\\
1 &\text{if }1\leq z_{k}<\lvert\Sigma_{k}\rvert,\\
0&\text{if }z_{k}=0.
\end{cases}$$
\end{itemize}
\end{definition}

We further define \textit{sum rule} $R^{SUM}$, \textit{sign sum rule} (\textit{S sum}) $R^{SSUM}$, \textit{ sum rule  with tie-breaking by lex-cel} $R^{SUM,L}$, and \textit{S sum rule  with tie-breaking by S lex-cel} $R^{SSUM,SL}$, as follows.
\begin{definition}
\label{definition:sum rules}
 For any $\succsim\in\domain$ with $\succsim:\Sigma_{1}\succ\cdots\Sigma_{l}$ and $x,y\in X$, 
\begin{align*}
xR^{SUM}_{\succsim}y&\iff \sum_{k\in[l]}x_{k}\geq \sum_{k\in[l]}y_{k};\\
xR^{SSUM}_{\succsim}y&\iff \sum_{k\in[l]}\dot{x}_{k}\geq \sum_{k\in[l]}\dot{y}_{k};\\
 xR^{SUM,L}_{\succsim}y&\iff (xP^{SUM}_{\succsim}y\text{ or }(xI^{SUM}_{\succsim}y\text{ and }xR^{L}_{\succsim}y));\\
xR^{SSUM,SL}_{\succsim}y&\iff (xP^{SSUM}_{\succsim}y\text{ or }(xI^{SSUM}_{\succsim}y\text{ and }xR^{SL}_{\succsim}y)).
\end{align*}
\end{definition}


Next, some variants of plurality are introduced. 
\begin{definition}
\label{definition:variants of plurality}
For any $\succsim\in\domain$ and $x,y\in X$, 
\begin{itemize}
\item \textit{Split plurality $R^{splitP}$} as follows. 

Let $\theta^{split}_{\succsim}(x):=(x^{split}_{1},\cdots,$ $x^{split}_{l})$, where 

$x^{split}_{k}:=\sum_{S\in\Sigma_{k}[x]}\frac{1}{\lvert S\rvert}$, and then

$xR^{split}_{\succsim}y\iff x_1^{split}\geq y_1^{split}$.
\item \textit{Plurality with tie-breaking by $\vartriangleright$}, denoted as $R^{P,\vartriangleright}$, as follows: 
$xR^{P,\vartriangleright}_{\succsim}y\iff xP^{P}_{\succsim}y\text{ or }(xI^{P}_{\succsim}y\text{ and }x\vartriangleright y)$.
\end{itemize}
\end{definition}

Finally, we define the following trivial SRS.

\begin{definition}
\label{definition:constant rule}
\textit{Constant-$X$} $R^{constX}$ is an SRS such that for any $\succsim\in\domain$, we have $R^{constX}=X\times X$.
\end{definition}

We are now ready to claim the logical independence of the axioms systems studied in Section \ref{sec:results} (for space reasons, a formal proof of each proposition is  provided in the Appendix).

\begin{proposition}\label{prop:indipSL1}
The axioms NT, CCON, AIAW, and IAWS in Theorem \ref{theorem:SL}.(1) are logically independent. In fact, \begin{itemize}
\item Inverse dual S lex-cel $R^{IDSL}$ satisfies the four axioms except IAWS; 
\item Lex-cel $R^{L}$ satisfies the four axioms except AIAW; 
\item S plurality $R^{SP}$ satisfies the four axioms except CCON; 
\item  S lex-cel with a particular attention to the non-existence in higher class $R^{SLNEH}$ satisfies the four axioms except NT.
\end{itemize}
\end{proposition}

\begin{proposition}\label{prop:indipSL2}
The axioms in NT, CCON, AIAW, and IDWS in Theorem \ref{theorem:SL}.(2) are logically independent. In fact, 
\begin{itemize}
\item S sum rule with
tie-breaking by S lex-cel $R^{SSUM,SL}$ satisfies the four axioms except IDWS;
\item S lex-cel with a particular precedence on unanimity $R^{SLUN}$ satisfies the four axioms except AIAW;
\item S plurality $R^{SP}$ satisfies the four axioms except CCON;
\item S lex-cel with a particular attention to the non-existence $R^{SLNE}$ satisfies the four axioms except NT.
\end{itemize}
\end{proposition}

\begin{proposition}\label{prop:indip-plurality}
The five axioms in Theorem \ref{theorem:plurality} are logically independent. In fact, 
\begin{itemize}
\item Lex-cel $R^{L}$ satisfies the five axioms except TO;
\item S plurality $R^{SP}$ satisfies the five axioms except TCON;
\item Constant-$X$ $R^{constX}$ satisfies the five axioms except WUVIP;
\item Split plurality $R^{splitP}$ satisfies the five axioms except WCA;
\item Plurality with tie-breaking by $\vartriangleright$ $R^{P,\vartriangleright}$ satisfies the five axioms except NT.
\end{itemize}
\end{proposition}

\begin{proposition}\label{prop:indipL}
The six axioms in Theorem \ref{theorem:L2} are logically independent.  In fact,
\begin{itemize}
\item  SUM rule with tie-breaking by Lex-cel $R^{SUM,L}$  satisfies the six axioms except IAWS,
\item S lex-cel $R^{SL}$ satisfies the six axioms except TCON;
\item Plurality $R^{P}$ satisfies the six axioms except CCON;
\item Constant-$X$ rule $R^{constX}$ satisfies the six axioms except WUVIP;
\item Split lex-cel $R^{splitL}$ satisfies the six axioms except WCA;
\item Lex-cel with tie-breaking by $\vartriangleright$ $R^{L,\vartriangleright}$ satisfies the six axioms except NT.
\end{itemize}
\end{proposition}

\section{Concluding remarks}\label{section:conclusion}

This paper introduces and investigates a series of consistency axioms (CCON, TCON, and BCON). These three axioms require that a social ranking obtained from two ``disjoint'' coalitional rankings $\succsim_{1}$ and $\succsim_{2}$ (e.g., the performance ranking of teams within a research department 1, and that of teams within a research department 2) should be compatible with the social ranking obtained from a particular ``sum'' of the two coalitional rankings (for instance, if the two departments are treated equally, in the ``sum,'' the $k^{\text{th}}$ best coalitions among department 1 are considered indifferent with the $k^{\text{th}}$ best coalitions among department 2; this is an example of what we called the top-aligned sum). 

Our main results from Theorem \ref{theorem:SL} (2$\Rightarrow$3), Theorem \ref{theorem:L2}, and Theorem \ref{theorem:plurality} can be summarized as follows: 
\begin{alignat*}{5}
R^{SL}:  &[\text{NT(+WCA+WUVIP)}]  &+& [\text{CCON}]           & &+& [\text{IDWS}]        &+& [\text{AIAW}], \\
R^{L}:  &[\text{NT+WCA+WUVIP}] &+& [\text{CCON}+\text{TCON}] & &+& [\text{IAWS}]       &  &              , \\
R^{P}:   &[\text{NT+WCA+WUVIP}] &+& [\text{TCON}]            &&+& [\text{TO}]          &  &              .
\end{alignat*}
We can observe that axioms used in the characterizations of $R^{SL}$, $R^{L}$, and $R^{P}$ (WCA and WUVIP for $R^{SL}$ are placed between two round brackets to indicate that they are implied by the other axioms) can be classified in four main  blocks with a common meaning. The first block (NT, WCA, and WUVIP) are the common basis of the three SRSs. The second block (CCON and/or TCON) describes the differences of the three SRSs in terms of consistencies: $R^{L}$ satisfies both CCON and TCON, while $R^{SL}$ satisfies only CCON, and $R^{P}$ satisfies only TCON (Proposition \ref{proposition:relationship of consistencies}). The third block (IDWS, IAWS, and TO) refers to the independence of the worst/non-top class(es), and the fourth block (AIAW) is the axiom specifying the winners when all the coalitions are indifferent under the given coalitional ranking. 

In summary, our new consistency axioms effectively  capture the differences of some SRSs, as shown by our axiomatic characterizations of three different SRSs, as well as by our results summarized in Figures \ref{figure:consistencies} and \ref{figure:concatenation consistencies}. Further investigation of the logical relationship between the consistency axioms in connection with other SRSs from the literature is a promising future research topic. 

\appendix
\section{Complementary SRSs}
\label{appendix:complementary SRSs}
Table \ref{table:index} summarizes the list of complementary SRSs that have been defined in the paper and used to prove the independence of the axioms in the next section.
\begin{table}[H]
\centering
\caption{Index of SRSs}
\begin{tabular}{{p{5.8cm} l l}}
\toprule
\textbf{Name}& \textbf{SRS} & \textbf{Def.}\\
\midrule
Anti-plurality & $R^{AP}$& Def. \ref{definition:pluralities}\\
Constant-$X$ & $R^{constX}$ & Def. \ref{definition:constant rule}\\
Dual S lex-cel & $R^{DSL}$& Def. \ref{definition:variants of lex-cel and S lex-cel}\\
IIS & $R^{IIS}$ & Def. \ref{definition:IIS}\\
Inverse dual S lex-cel & $R^{IDSL}$& Def. \ref{definition:variants of lex-cel and S lex-cel}\\
Lex-cel & $R^{L}$& Def. \ref{definition:lex-cel and S lex-cel} \\
Lex-cel with tie-breaking by $\vartriangleright$ & $R^{L,\vartriangleright}$ & Def. \ref{definition:variants of lex-cel and S lex-cel}\\
Plurality & $R^{P}$& Def. \ref{definition:pluralities}\\

Plurality with tie-breaking by $\vartriangleright$ & $R^{P,\vartriangleright}$ & Def. \ref{definition:variants of plurality}\\
S lex-cel & $R^{SL}$& Def. \ref{definition:lex-cel and S lex-cel} \\
S lex-cel with 
attention to the non-existence & $R^{SLNE}$ & Def. \ref{definition:further variants of lex-cel}\\
S lex-cel with 
attention to the non-existence in higher classes& $R^{SLNEH}$ & Def. \ref{definition:further variants of lex-cel}\\
S lex-cel with 
precedence on unanimity & $R^{SLUN}$ & Def. \ref{definition:further variants of lex-cel}\\
S plurality & $R^{SP}$& Def. \ref{definition:pluralities}\\
S sum rule  & $R^{SSUM}$ & Def. \ref{definition:sum rules}\\
S sum rule  with tie-breaking by S lex-cel & $R^{SSUM,SL}$ & Def. \ref{definition:sum rules}\\
Split lex-cel & $R^{splitL}$ & Def. \ref{definition:variants of lex-cel and S lex-cel}\\
Split Plurality & $R^{splitP}$& Def. \ref{definition:variants of plurality}\\
Sum rule & $R^{SUM}$ & Def. \ref{definition:sum rules}\\
Sum rule  with tie-breaking by lex-cel & $R^{SUM,L}$ & Def. \ref{definition:sum rules}\\
\bottomrule
\end{tabular}
\label{table:index}
\end{table}

\section{Independence of the axioms}

\setcounter{proposition}{4}

\begin{proof}[\bf Proof of Proposition \ref{prop:indipSL1}]
The proof is almost straightforward. 

\textbf{On $R^{IDSL}$}. It is straightforward to confirm that $R^{ISDL}$ satisfies NT, CCON, and AIAW---however, it does not satisfy IAWS. A counterexample is as follows: for $\succsim:\{\{x\}\}$ and $\succsim':\{\{x\}\}\succ'\{\{y\}\}$. Then, we have $xP^{IDSL}_{\succsim}y$ and $yP_{\succsim'}^{ISDL}x$, which contradict IAWS.

\textbf{On $R^{L}$}, the statement is almost straightforward. 

\textbf{On $R^{SP}$}. It is easy to confirm that $R^{SP}$ satisfies NT, AIAW, and IAWS---however, it does not satisfy CCON. A counterexample is as follows: let $\succsim:\{\{x,y\}\}$ and $\succsim':\{\{y\}\}$. Then, we have $xI_{\succsim}^{SP}y$, $xP_{\succsim'}^{SP}y$, and yet $xI^{SP}_{\succsim\cdot\succsim'}y$. This contradicts CCON (IP-CCON).

\textbf{On $R^{SLNEH}$}. It is straightforward to confirm that $R^{SLNEH}$ satisfies AIAW and IAWS, but not NT. We will prove that $R^{SLNEH}$ satisfies CCON. 
By definition of $R^{SLNEH}$, we can say that for any $\succsim\in\domain$ and $x,y\in X$, it follows that 
\begin{equation}
\label{equation:I^SLNEH}
xI^{SLNEH}_{\succsim}y\iff \dot{\theta}_{\succsim}(x)=\dot{\theta}_{\succsim}(y)=\boldsymbol{1}.
\end{equation}
Assume that $xI_{\succsim}^{1}y$ and $xI^{SLNEH}_{\succsim'}y$. By equation \eqref{equation:I^SLNEH}, we can say that $\dot{\theta}_{\succsim}(x)=\dot{\theta}_{\succsim}(y)=\boldsymbol{1}$ and $\dot{\theta}_{\succsim'}(x)=\dot{\theta}_{\succsim'}(y)=\boldsymbol{1}$. Hence, we note that $\dot{\theta}_{\succsim\cdot\succsim'}(x)=\dot{\theta}_{\succsim\cdot\succsim'}(y)=\boldsymbol{1}$. By equation \eqref{equation:I^SLNEH} again, we can conclude that $xI^{SLNEH}_{\succsim\cdot\succsim'}y$. This proves that $R^{SLNEH}$ satisfies II-CCON.

Assume that $xI_{\succsim}^{SLNEH}y$ and $xP_{\succsim'}^{SLNEH}y$. By equation \ref{equation:I^SLNEH}, we have $\dot{\theta}_{\succsim}(x)=\dot{\theta}_{\succsim'}(y)=\boldsymbol{1}$. If $xP_{\succsim'}^{SLNEH}y$ is obtained by step (a) (i.e., there exists $\hat{k}\in[l]$ such that $\dot{x}_{k}=\dot{y}_{k}=1$ for all $k<\hat{k}$ and $\dot{x}_{\hat{k}}=\dot{y}_{\hat{k}}=0$ and $x\vartriangleright y$), we can infer that (a) is also applied in $\succsim\cdot\succsim'$. Hence, we obtain that $xP^{SLNEH}_{\succsim\cdot\succsim'}y$. Next, if $xP_{\succsim'}^{SLNEH}y$ is obtained by step (b), we can infer that (b) is also applied in $\succsim\cdot\succsim'$. Hence, we obtain that $xP_{\succsim\cdot\succsim'}^{SLNEH}y$. In either case, we have verified that $xP^{SLNEH}_{\succsim\cdot\succsim'}y$. This proves that $R^{SLNEH}$ satisfies IP-CCON.

Similarly, PI-CCON and PP-CCON are shown. 
\end{proof}

\begin{proof}[\bf Proof of Proposition \ref{prop:indipSL2}]
Most of the proof is straightforward. We will prove only the non-trivial sections. 

\textbf{On $R^{SSUM,SL}$:} NT, and AIAW are straightforward. We can easily confirm that for any $\succsim\in\domain$ and $x,y\in X$, 
\begin{align}
\label{equation:I^{SSUM,SL}}
xI^{SSUM,SL}_{\succsim}y&\iff \dot{\theta}_{\succsim}(x)=\dot{\theta}_{\succsim}(y).
\end{align}
CCON is also straightforward by equation \ref{equation:I^{SSUM,SL}}. Furthermore, $R^{SSUM,SL}$ does not satisfy IDWS. A counterexample is as follows. Let $\succsim:\{\{x\}\}\succ\{\{y\},\{y,z\}\}$, and $\succsim':\{\{x\}\}\succ'\{\{y\}\}\succ'\{\{y,z\}\}$. Then, we have $xP^{SSUM,SL}_{\succsim}y$ and $yP^{SSUM,SL}_{\succsim'}x$. These contradict IDWS.\\


\textbf{On $R^{SLUN}$:}  CCON and NT are straightforward. We can verify that $R^{SLUN}$ satisfies IDWS as follows: Take $x,y\in X$ and $\succsim:\Sigma_1\succ\ldots\succ\Sigma_l$, and suppose $x P_\succsim^{SLUN} y$. If there exists $\hat{k}<l$ such that $\tilde{x}_{k}=\tilde{y}_{k}$ for all $k<\hat{k}$ and $\tilde{x}_{\hat{k}}>\tilde{y}_{\hat{k}}$, then the proof is straightforward (this is because the decomposition of the worst set would not change the first $\hat{k}$ equivalence classes). Suppose $\tilde{x}_{k}=\tilde{y}_{k}$ for all $k<l$ and $\tilde{x}_l>\tilde{y}_l$. Observe that, if $\tilde{y}_l=0$ then $x P_{\succsim'}^{SLUN} y$ , for any $\succsim'$ obtained by decomposing the worst equivalence class. If $\tilde{y}_l=1$ then $\tilde{x}_l=2$ and $x$ are contained in every set of $\Sigma_l$, then $x P_{\succsim'}^{SLUN} y$, with $\succsim'$ taken as before. So, $R^{SLUN}$ satisfies IDWS. Finally, $R^{SLUN}$ does not satisfy AIAW. A counterexample is $\succsim:\{\{x,y\},\{x\}\}$. Then, AIAW demands that $x$ and $y$ are judged indifferently, but $xP^{SLUN}_{\succsim}y$.   \\
\textbf{On $R^{SP}$:} As in the proof of Proposition \ref{prop:indipSL1}, CCON is not satisfied, while AIAW and NT hold. Moreover, by construction, IDWS is satisfied.\\
\textbf{On $R^{SLNE}$:} It is straightforward to confirm that $R^{SLNE}$ satisfies IDWS and AIAW, but not NT. Observe that \begin{align}
\label{equation:I^SLNE}
xI^{SLNE}_{\succsim}y&\iff \dot{\theta}_{\succsim}(x)=\dot{\theta}_{\succsim}(y)=\boldsymbol{1}.
\end{align}
Using equation \ref{equation:I^SLNE}, it is easy to prove that $R^{SLNE}$ satisfies CCON following the same steps in the proof of Proposition \ref{prop:indipSL1} for $R^{SLNEH}$. 

\end{proof}

\begin{proof}[\bf Proof of Proposition \ref{prop:indip-plurality}]
Most of the proof is straightforward. We will prove only the non-trivial sections. \\
\textbf{On $R^{L}$:} To observe that TO is not satisfied, take $\succsim:\{x,y\}\succ\mathfrak{X}\setminus\{\{x,y\}\}$ and $\succsim':\{x,y\}\succ'\{x\}\succ'\mathfrak{X}\setminus\big\{\{x\},\{x,y\}\big\}$. Then, we have $xI_{\succsim} y$ and $xP_{\succ'}y$, violating TO. It is immediately shown that NT ,WCA, WUVIP, and TCON are satisfied.\\
\textbf{On $R^{SP}$:} Take $\succsim_1:\{\{x\}\}\succsim\mathfrak{X}\setminus\{\{x\}\}$ and $\succsim_2:\{\{x,y\}\}\succsim\mathfrak{X}\setminus\{\{x,y\}\}$, then $x P_{\succsim_1}^{SP} y$ and $x I_{\succsim_2}^{SP} y$, but $x I_{\succsim_1\topsum\succsim_2}^{SP} y$ and TCON are not satisfied. It is straightforward to show that NT, WCA, WUVIP, and TO are satisfied.  \\
\textbf{On $R^{constX}$:} By definition, for any $x,y\in X$, $x I^{constX} y$ then WUVIP is not satisfied---on the other hand, NT, WCA,  TCON, and TO are satisfied.\\
\textbf{On $R^{splitP}$:} Take $X=\{a,b,c,x,y\}$ and $\succsim\in\domain$ with $\succsim:\Sigma_1=\big\{\{x,a\},\{y,b,c\}\big\}\succ \mathfrak{X}\setminus\Sigma_1$. Define a $\{x,y\}$-invariant permutation $\pi:\mathfrak{X}\to\mathfrak{X}$ such that $\pi(\{x,a\})=\{x,b,c\}$, $\pi(\{y,b,c\})=\{y,a\}$ and leaves the other sets unchanged, then $x P_\succsim^{splitP} y$, but $y P_{\succsim_\pi}^{splitP} x$. It is simple to show that NT, WUVIP, TCON, and TO are satisfied because the same arguments used for plurality $R^P$ can be used here.\\
\textbf{On $R^{P,\vartriangleright}$:} Take $x,y\in X$ such that $x\vartriangleright y$ and a permutation $\pi:X\to X$ such that $\pi(x)=y$, $\pi(y)=x$ and does not change the other elements. Take $\succsim\in\domain$ with $\succsim:\{x,y\}\succ \mathfrak{X}\setminus\{\{x,y\}\}$. By definition $x P_{\succsim}^{P,\vartriangleright} y$, but $\pi(y) P_{\succsim_\pi}^{P,\vartriangleright} \pi(x)$ and NT is not satisfied. WCA is satisfied because for any $\{x,y\}$-invariant permutation $\pi:\mathfrak{X}\to\mathfrak{X}$ if $x P_\succsim^P y$, then $x P_{\succsim_\pi}^P y$ and if $x I_\succsim^P y$, then $x I_{\succsim_\pi}^P y$. WUVIP follows directly by definition. By construction, we only have to check property (4) in Definition \ref{definition:consistencies}. If $x P_{\succsim}^{P,\vartriangleright} y$ and $x P_{\succsim'}^{P,\vartriangleright} y$, then in $\succsim''=\succsim\bottomsum\succsim'$ either $x P_{\succsim''}^{P} y$ or $x I_{\succsim''}^{P} y$ and then TCON is satisfied. TO is satisfied by construction.\\
\end{proof}

\begin{proof}[\bf Proof of Proposition \ref{prop:indipL}]
\textbf{On $R^{SUM,L}$:} Take $\succsim:\{\{x\}\}\succ\{\{x,y\}\}$ and $\succsim':\{\{x\}\}\succ\{\{x,y\}\}\succ\{\{y\}\}$,  then $x P_{\succsim}^{SUM,L} y$---however, $x I_{\succsim'}^{SUM, L} y$ and IAWS are not satisfied. It is straightforward to show that NT, WCA, WUVIP, TCON, and CCON are satisfied---by the definition of Lex-cel.   \\
\textbf{On $R^{SL}$:} The proof is analogous to the proof of Proposition \ref{prop:indip-plurality}. Axioms IAWS and CCON are satisfied by the definition of the lexicographic comparison of the vectors $\theta_{\succsim}(\cdot)$. \\
\textbf{On $R^{P}$:} By Proposition \ref{proposition:relationship of consistencies} CCON is not satisfied, and by Theorem \ref{theorem:plurality} TCON, WUVIP, WCA, and NT are satisfied. Axiom IAWS follows immediately by definition.\\
\textbf{On $R^{constX}$:} By construction, WUVIP is not satisfied and the other axioms hold.   \\
\textbf{On $R^{splitL}$:}  The proof is analogous to the proof of Proposition \ref{prop:indip-plurality} and WCA is not satisfied. Axioms CCON is satisfied because by taking two rankings---$\succsim_1$ and $\succsim_2$---the split sums do not change in $\succsim_1\cdot\succsim_2$. A similar argument applies to show that axiom IAWS hold. \\
\textbf{On $R^{L,\vartriangleright}$:} The proof is analogous to the proof of Proposition \ref{prop:indip-plurality}, and both CCON and IAWS follow from construction.  
\end{proof}

\section*{Acknowledgments}
Takahiro Suzuki was supported by JSPS KAKENHI Grant Number JP25K01302. Stefano Moretti acknowledges financial support from the ANR project THEMIS (ANR-20-CE23-0018). Michele Aleandri is a member of GNAMPA of the Istituto Nazionale di Alta Matematica (INdAM).


\bibliographystyle{cas-model2-names}

\bibliography{library}



\end{document}